\documentclass[conference]{IEEEtran}
\usepackage[top=0.75in,bottom=1.05in,left=0.65in,right=0.63in]{geometry}

\usepackage[T1]{fontenc}
\usepackage{amsmath,amssymb,amsfonts,mathrsfs}
\usepackage{amsthm}
\usepackage{bm}
\usepackage{cite}
\usepackage{graphicx}

\theoremstyle{plain}
\newtheorem{theorem}{Theorem}
\newtheorem{proposition}{Proposition}
\newtheorem{lemma}{Lemma}

\newtheorem{example}{Example}

\theoremstyle{definition}
\newtheorem{definition}{Definition}

\theoremstyle{remark}

\newcommand{\F}{\mathbb{F}}
\newcommand{\p}{\mathbf{p}}
\newcommand{\Enc}{\mathrm{Enc}}
\newcommand{\wt}{\mathrm{wt}}
\newcommand{\dH}{d_H}

\title{Function-Correcting Codes with Optimal Data Protection for Hamming Code Membership}

\author{
\IEEEauthorblockN{Swaraj Sharma Durgi\IEEEauthorrefmark{1}, Anjana A. Mahesh\IEEEauthorrefmark{1}, Anupriya Kumari\IEEEauthorrefmark{2}, Rajlaxmi Pandey\IEEEauthorrefmark{3}, B. Sundar Rajan\IEEEauthorrefmark{3}}
\IEEEauthorblockA{\IEEEauthorrefmark{1}Department of Electrical Engineering, IIT Hyderabad, India}
\IEEEauthorblockA{\IEEEauthorrefmark{2}Department of Electronics and Communication Engineering, NIT Patna, India}
\IEEEauthorblockA{\IEEEauthorrefmark{3}Department of Electrical Communication Engineering, IISc Bengaluru, India}
\IEEEauthorrefmark{1}{\{ee24btech11018@iith.ac.in, anjana.am@ee.iith.ac.in\}},\IEEEauthorrefmark{2}{anupriyak.ug22.ec@nitp.ac.in}, \IEEEauthorrefmark{3}{\{rajlaxmip, bsrajan\}@iisc.ac.in}
}

\begin{document}
\maketitle

\begin{abstract}
This paper investigates single-error-correcting function-correcting codes (SEFCCs) for the Hamming code membership function (HCMF), which indicates whether a vector in \(\mathbb{F}_2^7\) belongs to the \([7,4,3]\)-Hamming code. Necessary and sufficient conditions for valid parity assignments are established in terms of distance constraints between codewords and their nearest non-codewords. It is shown that the Hamming-distance-\(3\) relations among Hamming codewords induce a bipartite graph, a fundamental geometric property that is exploited to develop a systematic SEFCC construction. By deriving a tight upper bound on the sum of pairwise distances, we prove that the proposed bipartite construction uniquely achieves the maximum sum-distance, the largest possible minimum distance of $2$, and the minimum number of distance-2 codeword pairs. Consequently, for the HCMF SEFCC problem, sum-distance maximization is not merely heuristic—it exactly enforces the optimal distance-spectrum properties relevant to error probability. Simulation results over AWGN channels with soft-decision decoding confirm that the resulting max-sum SEFCCs provide significantly improved data protection and Bit Error Rate (BER) performance compared to arbitrary valid assignments.

\end{abstract}

\begin{IEEEkeywords}
Function-correcting codes, Hamming codes, membership function, data protection, soft-decision decoding.
\end{IEEEkeywords}

\section{Introduction}

In many modern communication and distributed computing systems, the receiver's objective is not the full reconstruction of the transmitted message, but rather the computation of a specific function of that message. Classical error-correcting codes (ECCs), designed to protect the entire message, often provide unnecessary redundancy for such tasks. This observation led to the development of function-correcting codes (FCCs), introduced by Lenz \textit{et al.} \cite{lenz2023fcc}, which guarantee that a function $f(\mathbf{u})$ can be correctly computed even if the codeword representing the message $\mathbf{u}$ is corrupted by up to $t$ errors.

Since the formalization of FCCs, significant research has focused on characterizing the fundamental limits of redundancy and constructing codes for specific channel models. The optimal redundancy for various function classes was explored in \cite{ge2025optimal} and \cite{ly2025finitefields}. Extensions to specialized read channels, such as symbol-pair and $b$-symbol channels, were investigated in \cite{xia2024symbolpair, singh2025bsymbol, sampath2025note}. Furthermore, structural properties of FCCs, including homogeneous distance requirements and locally bounded function constraints, have been recently studied in \cite{liu2025homogeneous, rajput2025locally}. Recent work by Rajput \textit{et al.} has also introduced the concepts of specialized partition codes \cite{rajput2025fccPC} and function-correcting codes with data protection \cite{rajput2025fccdp}.

Most existing FCC literature treats the function $f$ as a general mapping or focuses on functions with specific local properties. There is a lack of study concerning functions defined by algebraic membership in a subspace of a vector space. In this paper, we address this gap by studying FCCs for the Hamming-code membership function (HCMF), denoted $f_H$. The function $f_H:\F_2^7 \to \F_2$ indicates whether a given vector belongs to the $[7,4,3]$-Hamming code, viewed as a $4$-dimensional linear subspace of $\F_2^7$.  

As discussed in \cite{rajput2025fccdp}, multiple FCCs may exist with the same redundancy and the same function-correcting capability, yet they can offer varying levels of protection for the underlying data. This implies that error performance with respect to data bits is not solely determined by the function-correction requirement. Our work is motivated by a similar premise as in  \cite{rajput2025fccdp} but differs significantly in its objective and scope. While \cite{rajput2025fccdp} provides a general framework for adding a fixed level of data protection to an FCC, in this paper, we consider optimal single error-correcting FCCs (SEFCCs) for the HCMF, analyze their distance properties, and identify a class of SEFCCs that has superior data bit error performance.

The contributions of this paper are summarized as follows:
\begin{itemize}
    \item Necessary and sufficient conditions are established for a parity assignment to constitute a valid SEFCC for the HCMF.
    \item The set of Hamming codewords at minimum distance is shown to induce a bipartite graph, and this bipartite structure is identified as the key geometric property required to satisfy FCC validity constraints.
    \item A systematic SEFCC construction is developed using this bipartite framework, achieving both the maximum theoretical total pairwise distance and the largest possible minimum distance.
    \item We show that a balanced parity split within each partite set minimizes the number of low-distance codeword pairs via a quadratic optimization of the distance spectrum.
    \item We establish that for the HCMF SEFCC problem, sum-distance maximization is not merely heuristic—it exactly enforces the optimal distance-spectrum properties relevant to error probability.
    \item AWGN channel simulations demonstrate that the proposed construction yields a significant Bit Error Rate improvement over arbitrary SEFCCs.
\end{itemize}


The rest of the paper is organized as follows. Section \ref{sec:Prelims} provides the necessary preliminaries on function-correcting codes and the $[7,4,3]$-Hamming code. Section \ref{sec:Setting} defines the problem setup and the Hamming code membership function. Section \ref{sec:Results} presents our main results: we establish the necessary and sufficient conditions for valid parity assignment, prove the bipartite nature of the distance-3 codeword graph, and propose the optimal max-sum construction while deriving its distance properties. Section \ref{sec:Sim} presents the performance evaluation via AWGN channel simulations, and Section \ref{sec:Conc} concludes the paper.

\textit{Notations:} For any positive integer \( n \), the shorthand \( [n] \) refers to the set \( \{1, 2, \dots, n\} \). For a real number $x$, the notation $[x]^+$ is used to denote $\max(0,x)$. The finite field with \( 2 \) elements is denoted by \( \mathbb{F}_2 \). Vectors of length \( n \) over \( \mathbb{F}_2 \) form the vector space \( \mathbb{F}_2^n \). For a binary vector $\mathbf{x}$, its bitwise complement is denoted as $\overline{\mathbf{x}}$. A linear block code with length \( n \), dimension \( k \), and minimum Hamming distance \( d \) over \( \mathbb{F}_2 \) is denoted as a \( [n, k, d] \) code. The Hamming weight of a vector \( \mathbf{x} \in \mathbb{F}_2^n \), denoted \( \wt(\mathbf{x}) \), is the number of positions in \( \mathbf{x} \) that are nonzero. The Hamming distance between two vectors or \( \mathbf{x}, \ \mathbf{y} \in \mathbb{F}_2^n \), denoted as \( \dH(\mathbf{x},\mathbf{y})\) is the number of positions where the bits in $\mathbf{x}$ and $\mathbf{y}$ differ. 
\section{Preliminaries}
\label{sec:Prelims}

\subsection{Function-Correcting Codes}
\begin{definition}[(Systematic) Function-Correcting Code]\cite{lenz2023fcc}
Let $f:\F_2^k\to \mathscr{I}$ be a function.
A systematic encoder $\Enc:\F_2^k\to \F_2^{k+r}$ of the form $\Enc(\mathbf{u})=(\mathbf{u},\p(\mathbf{u}))$ is an $(f,t)$-FCC if for all $\mathbf{u}_1,\mathbf{u}_2\in \F_2^k$ with $f(\mathbf{u}_1)\neq f(\mathbf{u}_2)$,
\begin{equation}
\dH\big(\Enc(\mathbf{u}_1),\Enc(\mathbf{u}_2)\big)\ge 2t+1.
\label{eq:fccdef}
\end{equation}
\end{definition}

\begin{definition}[Optimal redundancy]\cite{lenz2023fcc}
Given $f$ and $t$, the optimal redundancy $r_f(k,t)$ is the minimum $r$ such that there exists an encoder $\Enc:\F_2^k\to \F_2^{k+r}$ that forms an $(f,t)$-FCC.
\end{definition}


\subsection{$[7,4,3]$-Hamming Code}
\label{sec:HC}
The $[7,4,3]$-Hamming code is a binary linear block code that encodes $4$ information bits into $7$ coded bits and has minimum Hamming distance $3$. Being a perfect code, meaning that every vector in $\F_2^7$ lies within Hamming distance one of a unique codeword, it partitions the space $\F_2^7$ into $16$ Hamming spheres, denoted as ${S_i}, i \in [16]$, with one of the 16 codewords, denoted as $\mathbf{c}_i$, at the centre of each sphere and $7$ unique non-codewords at a Hamming distance of $1$ from the codeword at the centre. We denote the non-codewords that are nearest to the codeword $\mathbf{c}_i$ as $\mathbf{u}_{i,j}, \ j \in [7]$. Thus, the vector space $\F_2^7$ is partitioned into $ \mathscr{C}_H= \{\mathbf{c}_i\}_{i=1}^{16}$ and $\F_2^{7}\setminus\mathscr{C}_H = \{\mathbf{u}_{i,j}: i \in [16], \ j \in [7]\}$. 

Owing to its algebraic structure, uniform distance properties, and optimality for single-error correction, the $[7,4,3]$-Hamming code serves as a canonical example in coding theory and provides a natural testbed for studying FCCs based on code membership functions~\cite{macwilliams1977theory}.

\section{Problem Setup}
\label{sec:Setting}
For the $[7,4,3]$-Hamming code defined in section \ref{sec:HC}, let $\mathscr{C}_H$ denote the set of 16 codewords. We study function correcting codes for the \emph{Hamming code membership function} (HCMF) $f_H:\F_2^7\to \F_2$ defined as 
\[
f_H(\mathbf{v})=
\begin{cases}
1, & \mathbf{v}\in \mathscr{C}_H,\\
0, & \mathbf{v}\in \mathbb{F}_2^7 \setminus \mathscr{C}_H,
\end{cases}
\]
Thus, the domain $\F_2^7$ contains $16$ Hamming codewords (HCWs) for which $f_H$ evaluates to 1 and $112$ vectors that are not Hamming codewords (NHCWs) for which $f_H$ evaluates to 0. 

For all boolean functions, it was shown in \cite{lenz2023fcc} that optimal redundancy for a $t$-error correcting FCC is $r_f(k,t) = 2t$. Hence, an optimal $t$-error correcting systematic FCC $\mathcal{C}$ for the Hamming code membership function $f_H$ is a mapping from $\mathbb{F}_2^7 \rightarrow \mathbb{F}_2^{7+2t}$ with $\Enc(\mathbf{u}_i) = [\mathbf{u}_i, \p_\mathcal{C}(\mathbf{u}_i)]$, where $\p_\mathcal{C}(\mathbf{u}_i) \in \mathbb{F}_2^{2t}$ is the $2t$-length parity vector assigned to the input vector $\mathbf{u}_i \in \mathbb{F}_2^7$. In this paper, we consider single-error-correcting function correcting codes (SEFCCs) for the function $f_H$. 
For single-error correction ($t=1$), the FCC distance constraint in \eqref{eq:fccdef} reduces to 
\begin{equation}
\dH\big(\Enc(\mathbf{c}_i),\Enc(\mathbf{u}_j)\big)\ge 3, \forall \mathbf{c}_i \in \mathscr{C}_H \text{ and } \mathbf{u}_j \in \mathbb{F}_2^7 \setminus \mathscr{C}_H.
\label{eq:cw_ncw_req}
\end{equation}


\begin{definition}[Hamming Distance matrix]
For an FCC $\mathcal{C}$ for a function $f:\F_2^k\to \mathscr{I}$, we define the Hamming distance matrix (HDM) of the FCC $\mathcal{C}$ as the $2^k \times 2^k$ matrix $D_\mathcal{C}$ with the $(i,j)^{\text{th}}$ entry defined as $$D_\mathcal{C}(i,j)= \dH(\mathbf{u}_i,\mathbf{u}_j)+ \dH(\p_\mathcal{C}(\mathbf{u}_i), \p_\mathcal{C}(\mathbf{u}_j)),$$ for $i,j \in [2^k]$, and $\mathbf{u}_i$ and $\mathbf{u}_j$ are the $i^{\text{th}}$  and $j^{\text{th}}$ elements in the natural binary counting order of the vectors in $\mathbb{F}_2^k$.  \end{definition}

\begin{definition}[Sum-Distance]
For an FCC $\mathcal{C}$ for a function $f:\F_2^k\to \mathscr{I}$ with a Hamming distance matrix $D_\mathcal{C}$, the sum-distance of $\mathcal{C}$, denoted $\Sigma_{d}(\mathcal{C})$, is defined as the sum of pairwise Hamming distances of the encodings in $\mathcal{C}$, i.e., $\Sigma_{d}(\mathcal{C}) = \sum_{i,j}D_\mathcal{C}(i,j)$. 
    
\end{definition}
\section{Main Results}
\label{sec:Results}

This section starts with a theorem on FCC construction for any binary-valued function that minimizes Function Error Rate (FER), then studies optimal SEFCCs for the HCMF \( f_H : \mathbb{F}_2^7 \to \mathbb{F}_2 \). We characterize distance properties that govern data bit error performance and construct a family of optimal SEFCCs achieving these properties.

\begin{theorem}[Optimal FER FCC]
 \label{thm:optimFER}
    For any binary-valued function $f:\F_2^k \rightarrow \F_2$, define $f^{-1}(0) \triangleq  \{\mathbf{u} \in \F_2^k \mid f(\mathbf{u})=0\}$ and $f^{-1}(1) \triangleq \{\mathbf{u} \in \F_2^k \mid f(\mathbf{u})=1\}$. An optimal FCC $\mathcal{C}$ that assigns the same parity $\p \in \F_2^2$ to all vectors in $f^{-1}(0)$ and the complement $\overline{\p}$ for all the vectors in $f^{-1}(1)$ achieves the optimal FER performance. 
\end{theorem}
\begin{IEEEproof}
    Since function errors arise exclusively from pairs $(\mathbf{u},\mathbf{v})$ with $f(\mathbf{u})\neq f(\mathbf{v})$, which we refer to as cross-class codeword pairs, minimizing function error reduces to maximizing the distances between all cross-class pairs. The maximum possible parity contribution to the distance of any cross-class pair is $2$. Since the proposed parity assignment  enforces $\dH(\p_{\mathcal{C}}(\mathbf{u}),\p_{\mathcal{C}}(\mathbf{v}))=2$ for every pair $(\mathbf{u},\mathbf{v})$  with $f(\mathbf{u})\neq f(\mathbf{v})$, and hence maximizes the distance of all such pairs simultaneously, it achieves the optimal FER performance.
\end{IEEEproof}

\subsection{Parity Assignment and FCC Validity}
\begin{proposition}[Complementary Parity Rule]
\label{prop:CW_NCW CPR}
    An optimal SEFCC $\mathcal{C}: \F_2^7 \rightarrow \F_2^9$ for the HCMF $f_H$ assigns parities that are bitwise complements of each other to every Hamming codeword and its $7$ NHCW neighbors, i.e., $\p_\mathcal{C}(\mathbf{u}_{i,j}) = \overline{\p_\mathcal{C}(\mathbf{c}_i)}$, $\forall j \in [7]$, for each $i \in [16]$. 
\end{proposition}
\begin{IEEEproof}
    An SEFCC for the HCMF $f_H$ has to satisfy the distance criterion in \eqref{eq:cw_ncw_req} for every (HCW, NHCW) pair. In particular, the distance requirement has to be met for every Hamming CW $\mathbf{c}_i$ and the $7$ non-codewords in its $1$-neighbourhood, i.e., $\dH(\Enc(\mathbf{c}_i),\Enc(\mathbf{u}_{i,j})) \geq 3$, $\forall j \in [7]$, for each $i \in [16]$. For the (HCW, NHCW) pairs within a $1$-Hamming sphere, this condition effectively implies that the distance between their parities is $2$, i.e., $\dH(\p_\mathcal{C}(\mathbf{c}_i),\p_{\mathcal{C}}(\mathbf{u}_{i,j})) = 2$. Since each parity vector is of length $2$, this requirement can be achieved only if the parities are bitwise complements of each other. 
\end{IEEEproof}

Thus, the design of an optimal SEFCC for the HCMF $f_H$ reduces to assigning parity vectors for the $16$ vectors in $\mathscr{C}_H$. After enforcing the NCW parity rule, violations can only occur via CW--NCW pairs that are at distance $2$. The following condition exactly characterizes validity for SEFCCs.

\begin{proposition}[Codeword-Neighbour Parity Rule]
\label{prop:CW_NPR}
A parity assignment for the $16$ Hamming codewords $\{\mathbf{c}_i\}_{i=1}^{16}$ results in a valid optimal SEFCC $\mathcal{C}: \F_2^7 \rightarrow \F_2^9$ for the HCMF $f_H$ if and only if no pair of Hamming codewords at distance $3$ is assigned complementary parities. 
\end{proposition}
\begin{IEEEproof}
    The theorem states that for any pair of HCWs $(\mathbf{c}_i, \mathbf{c}_j), \ i,j \in [16], i \neq j,$ s.t. $\dH(\mathbf{c}_i, \mathbf{c}_j)=3$, $\p_\mathcal{C}(\mathbf{c}_i) \neq \overline{\p_\mathcal{C}(\mathbf{c}_j)}$, in a valid SEFCC. 
    A parity assignment $\mathcal{C}$ satisfying the complementary parity rule in Proposition \ref{prop:CW_NCW CPR} violates the distance requirement in \eqref{eq:cw_ncw_req} only if a HCW $\mathbf{c}_i$ and a NHCW $\mathbf{u}_{k,j}$ that are at a Hamming distance $2$ are assigned the same two-bit parity vector. For a given HCW $\mathbf{c}_i$, all NHCW vectors $\mathbf{u}_{k,j}$ s.t. $\dH(\mathbf{c}_i,\mathbf{u}_{k,j}) = 2$ lie in $1$-Hamming spheres of Hamming codewords $\mathbf{c}_k$ such that the distance $\dH(\mathbf{c}_i,\mathbf{c}_{k}) = 3$.  Since, from Proposition \ref{prop:CW_NCW CPR}, we have $\p_\mathcal{C}(\mathbf{u}_{k,j}) = \overline{\p_\mathcal{C}(\mathbf{c}_k)}$, this implies that a parity assignment is invalid iff $\p_\mathcal{C}(\mathbf{c}_{i}) = \overline{\p_\mathcal{C}(\mathbf{c}_k)}$ for any pair  $(\mathbf{c}_i,\mathbf{c}_{k})$ of HCWs at Hamming distance $3$. 
\end{IEEEproof}

\begin{theorem}[Valid SEFCC]
\label{thm:validity}
An optimal SEFCC $\mathcal{C}: \F_2^7 \rightarrow \F_2^9$ for the HCMF $f_H$ must assign parities to the vectors in $\F_2^7$ satisfying the following two conditions: 
\begin{enumerate}
    \item[a)] $\p_\mathcal{C}(\mathbf{u}_{i,j}) = \overline{\p_\mathcal{C}(\mathbf{c}_i)}$, $\forall j \in [7]$, for each $i \in [16]$. 
    \item[b)] $\p_\mathcal{C}(\mathbf{c}_{i}) \neq \overline{\p_\mathcal{C}(\mathbf{c}_k)}$ for any pair  $(\mathbf{c}_i,\mathbf{c}_{k})$  such that $\dH(\mathbf{c}_i,\mathbf{c}_{k}) = 3$, $i,k \in [16]$. 
\end{enumerate}
\end{theorem}

\begin{IEEEproof}
Follows from Propositions \ref{prop:CW_NCW CPR} and \ref{prop:CW_NPR}.
\end{IEEEproof}

\subsection{Distance Optimization Among Valid FCCs}
While an SEFCC for the HCMF constructed according to Theorem~\ref{thm:optimFER} achieves optimal function error rate (FER), such constructions have minimum Hamming distance $d_{\min}=1$, which leads to poor data-bit error performance. This motivates restricting attention to the subclass of valid SEFCCs with $d_{\min}=2$, for which a necessary and sufficient condition is provided in Theorem~\ref{thm:dmin2}. Among this restricted class, we optimize the distance spectrum by maximizing the sum of pairwise codeword distances $\Sigma_d(\mathcal{C})$, corresponding to the first moment of the distance distribution. Although error probability under maximum-likelihood decoding depends on the full distance spectrum, we show that, for SEFCCs designed for the HCMF, sum-distance maximization is not merely heuristic but an exact optimality criterion: it simultaneously maximizes $d_{\min}$ and minimizes the number of minimum-distance codeword pairs. Consequently, this approach yields SEFCCs with optimal data-bit error performance among all valid constructions.


\begin{theorem}[Condition for $d_{\min}=2$]
\label{thm:dmin2}
An optimal SEFCC $\mathcal{C}: \F_2^7 \rightarrow \F_2^9$ for the HCMF $f_H$ will have a minimum distance $d_{\min}(\mathcal{C}) = 2$ if and only if no pair of HCWs at a distance $3$ are assigned the same parity vectors, i.e., $\p_\mathcal{C}(\mathbf{c}_{i}) \neq \p_\mathcal{C}(\mathbf{c}_k)$ for any pair  $(\mathbf{c}_i,\mathbf{c}_{k})$  such that $\dH(\mathbf{c}_i,\mathbf{c}_{k}) = 3$, $i,k \in [16]$. 
\end{theorem}

\begin{IEEEproof}
The minimum distance of the code $\mathcal{C}$ is given by $d_{\min}(\mathcal{C}) = \min_{\mathbf{u}\neq \mathbf{v}} \Big(\dH(\Enc(\mathbf{u}), \Enc(\mathbf{v})\Big)$ for $\mathbf{u}, \mathbf{v}  \in \F_2^7$. For any optimal SEFCC $\mathcal{C}$ of the HCMF $f_H$, the minimum distance is at most two, i.e., $d_{\min}(\mathcal{C}) \leq 2$. To achieve $d_{\min}(\mathcal{C})=2$, we must ensure that no pair $(\mathbf{u}, \mathbf{v})$ of vectors in $\F_2^7$ has their corresponding codewords at a distance of $1$.

For any HCW $\mathbf{c}_i$, the set of $7$ nearest non-codewords $\{\mathbf{u}_{i,j}\}_{j=1}^7$ all receive the same parity $\overline{\p_{\mathcal{C}}(\mathbf{c}_i)}$ by Proposition~\ref{prop:CW_NCW CPR}. Since the distance between distinct NHCWs in the same sphere is $\dH(\mathbf{u}_{i,j}, \mathbf{u}_{i,k}) = 2$, the pairwise Hamming distance between their encodings is $2 + \dH(\overline{\p_{\mathcal{C}}(\mathbf{c}_i)}, \overline{\p_{\mathcal{C}}(\mathbf{c}_i)}) = 2$. 

From Theorem~\ref{thm:validity}, the distance between the codewords, in any valid SEFCC $\mathcal{C}$ of the HCMF $f_C$, of any HCW and NHCW is at least $3$. Therefore, a distance of $1$ in $\mathcal{C}$ can occur only between the encodings of two NHCWs $\mathbf{u} \in S_i$ and $\mathbf{v} \in S_k$ from different spheres. For $\dH(\Enc(\mathbf{u}), \Enc(\mathbf{v})) = 1$, two conditions must hold simultaneously:
\begin{enumerate}
    \item $\dH(\mathbf{u}, \mathbf{v}) = 1$: This occurs if and only if the spheres $S_i$ and $S_k$ are adjacent, which implies $\dH(\mathbf{c}_i, \mathbf{c}_k) = 3$.
    \item $\p_{\mathcal{C}}(\mathbf{u}) = \p_{\mathcal{C}}(\mathbf{v})$: By Proposition~\ref{prop:CW_NCW CPR}, this is equivalent to $\overline{\p_{\mathcal{C}}(\mathbf{c}_i)} = \overline{\p_{\mathcal{C}}(\mathbf{c}_k)}$, or simply $\p_{\mathcal{C}}(\mathbf{c}_i) = \p_{\mathcal{C}}(\mathbf{c}_k)$.
\end{enumerate}

Consequently, $d_{\min}(\mathcal{C}) = 1$ iff there exists a pair of distance-3 HCW vectors that are assigned the same parity in $\mathcal{C}$. Avoiding such parity assignments ensures that all inter-sphere NHCW pairs have their codewords in $\mathcal{C}$ at a distance $\geq 1 + 1 = 2$. Since there are no other pairs of vectors that can have a pairwise distance of $2$ between their encodings in a valid SEFCC, this results in $d_{\min}(\mathcal{C})=2$.  
\end{IEEEproof}

\begin{lemma}
    \label{lem:maxsum}
    The sum-distance $\Sigma_{d}(\mathcal{C})$ for any optimal SEFCC $\mathcal{C}$ for the HCMF $f_H$ is upper bounded as $\Sigma_{d}(\mathcal{C}) \le 73728$.
\end{lemma}

\begin{IEEEproof}
 We have the sum-distance $\Sigma_d(\mathcal{C}) = \sum_{i=1}^{2^7} \sum_{j=1}^{2^7} \left( \dH(\mathbf{u}_i, \mathbf{u}_j) + \dH(\p_{\mathcal{C}}(\mathbf{u}_i), \p_{\mathcal{C}}(\mathbf{u}_j)) \right)$. The first term represents the total pairwise distance between all vectors in the domain $\mathbb{F}_2^7$. For $k=7$, this sum is constant: $\sum_{i=1}^{2^k} \sum_{j=1}^{2^k} \dH(\mathbf{u}_i, \mathbf{u}_j) = k 2^{2k-1} = 7 \cdot 2^{13} = 57344$.

The second term depends on the parity assignment $\p_{\mathcal{C}}(\mathbf{u})$. For each parity coordinate $j \in \{1, 2\}$, let $n_0^{(j)}$ and $n_1^{(j)}$ denote the number of vectors $\mathbf{u} \in \mathbb{F}_2^7$ such that the $j$-th bit of $\p_{\mathcal{C}}(\mathbf{u})$ is $0$ and $1$, respectively. The sum of parity distances is maximized when the bits are balanced ($n_0^{(j)} = n_1^{(j)} = 64$): $\sum_{i=1}^{2^7} \sum_{j=1}^{2^7} \dH(\p_{\mathcal{C}}(\mathbf{u}_i), \p_{\mathcal{C}}(\mathbf{u}_j)) = 2 \sum_{j=1}^{2} n_0^{(j)} n_1^{(j)} \leq 2 \cdot (2 \cdot 64 \cdot 64) = 16384$.
Summing these components, we obtain:
$\Sigma_d(\mathcal{C}) \leq 57344 + 16384 = 73728$.
\end{IEEEproof}


\subsection{Optimal SEFCC Construction}
\label{subsec:CodeConst}
The $[7,4,3]$-Hamming Code has a highly regular neighborhood structure. In particular, distance-$3$ relations among codewords form a graph with strong symmetry. We make use of this structure (Proposition \ref{prop:bipartite}) to construct optimal SEFCCs for the Hamming code membership function with maximum sum-distance property. We use the concise notation $u_1u_2$ to denote a 2-length vector $[u_1, u_2]$ in the rest of this manuscript, primarily to denote parity vectors.

\begin{proposition}[Bipartite structure of distance-3 HCW graph]
    \label{prop:bipartite}
Let $G_3 = (\mathcal{V}, \mathcal{E})$ be the graph whose vertices are the $16$ codewords of a $[7,4,3]$-Hamming code, with an edge between two vertices representing distinct Hamming codewords if they are at a Hamming distance $3$, i.e., $\mathcal{V} = \mathscr{C}_H$ and $\mathcal{E} = \{(\mathbf{c}_i,\mathbf{c}_j) \ | \ \dH(\mathbf{c}_i,\mathbf{c}_j)=3\}$. Then, $G_3$ is bipartite.
\end{proposition}

\begin{IEEEproof}
We know that for a linear code, $\dH(\mathbf{c}_1,\mathbf{c}_2)=\wt(\mathbf{c}_1\oplus \mathbf{c}_2)$.
If $\dH(\mathbf{c}_1,\mathbf{c}_2)=3$, then $\wt(\mathbf{c}_1\oplus \mathbf{c}_2)$ is odd, which implies $\wt(\mathbf{c}_1)$ and $\wt(\mathbf{c}_2)$ have opposite parity.
Thus, partitioning the codewords in a $[7,4,3]$-Hamming code by even/odd weight yields a bipartition.
\end{IEEEproof}

\subsubsection{FCC Construction Using Bipartition}
\label{subsubsec:Construction}
Let $\mathscr{C}_H \subset \F_2^7$ denote the set of Hamming codewords corresponding to a $[7,4,3]$-Hamming code. By Proposition~\ref{prop:bipartite}, $\mathscr{C}_H$
admits a bipartition
\(
\mathscr{C}_H = \mathcal{C}_o \cup \mathcal{C}_e,
\)
where $\mathcal{C}_o$ and $\mathcal{C}_e$ denote the subsets of Hamming codewords
with odd and even Hamming weights, respectively, and
$\lvert \mathcal{C}_o \rvert = \lvert \mathcal{C}_e \rvert = 8$.

Define two complementary parity-pair sets
\(
P_1 \triangleq \{00, \  11\}\) and 
\(P_2 \triangleq \{01, \ 10\}.
\)
The SEFCC $\mathcal{C}:\F_2^7 \to \F_2^9$ is constructed by appending two parity bits to each Hamming codeword according to the following steps.

\begin{enumerate}
    \item \textbf{Group assignment} (validity):  
    Assign parity vectors from $P_1$ to all Hamming codewords in $\mathcal{C}_o$, and parity
    vectors from $P_2$ to all Hamming codewords in $\mathcal{C}_e$.
    (Equivalently, the roles of $P_1$ and $P_2$ may be interchanged.)

    \item \textbf{Internal split} (maximum sum-distance):  
    Enforce an equal split of parity pairs within each group. Specifically,
    \begin{itemize}
        \item assign the parity vector $00$ to exactly four HCWs in
        $\mathcal{C}_o$ and $11$ to the remaining four HCWs in
        $\mathcal{C}_o$;
        \item assign the parity pair $01$ to exactly four HCWs in
        $\mathcal{C}_e$ and $10$ to the remaining four HCWs in
        $\mathcal{C}_e$.
    \end{itemize}
\end{enumerate}

\begin{example}
\label{ex:HC_code}
Consider the $[7,4,3]$-Hamming code $\left\{\textbf{c}_i = [\mathbf{m}_i, p_{i,1}, p_{i,2}, p_{i,3}], \ i \in [16]\right\}$, where the message vector $\mathbf{m}_i \in \mathbb{F}_2^4$, which is the $i^{\text{th}}$ vector in the natural binary counting order, results in parity bits $p_{i,1} = m_{i,2} \oplus m_{i,3} \oplus m_{i,4}$, $p_{i,2} = m_{i,1} \oplus m_{i,3} \oplus m_{i,4}$, and $p_{i,3} = m_{i,1} \oplus m_{i,2} \oplus m_{i,4}$. Hamming codewords partition into odd-weight $\mathcal{C}_o = \{\mathbf{c}_3, \mathbf{c}_4, \mathbf{c}_5, \mathbf{c}_6, \mathbf{c}_9, \mathbf{c}_{10}, \mathbf{c}_{15}, \mathbf{c}_{16}\}$ and even-weight $\mathcal{C}_e = \{\mathbf{c}_1, \mathbf{c}_2, \mathbf{c}_7, \mathbf{c}_8, \mathbf{c}_{11}, \mathbf{c}_{12}, \mathbf{c}_{13}, \mathbf{c}_{14}\}$ partite sets. An optimal bipartition-based SEFCC assigns the parity $00$ to every HCW in 
$\{\mathbf{c}_3, \mathbf{c}_4, \mathbf{c}_5, \mathbf{c}_6\}$, the parity $11$ to 
$\{\mathbf{c}_9, \mathbf{c}_{10}, \mathbf{c}_{15}, \mathbf{c}_{16}\}$, the parity $01$ to 
$\{\mathbf{c}_1, \mathbf{c}_2, \mathbf{c}_7, \mathbf{c}_8\}$, and the parity $10$ to 
$\{\mathbf{c}_{11}, \mathbf{c}_{12}, \mathbf{c}_{13}, \mathbf{c}_{14}\}$.
\end{example}



\subsubsection{Distance Properties of the Constructed FCCs}

\begin{theorem}[Max-sum FCCs]
\label{thm:maxsumFCC}
    All optimal SEFCCs $\mathcal{C}: \F_2^7 \rightarrow \F_2^9$ for the HCMF $f_H$ that attain the maximum sum-distance established in Lemma \ref{lem:maxsum} are obtained via the bipartition-based code construction described in section \ref{subsubsec:Construction}. 
    \end{theorem}

\begin{IEEEproof}
    A code $\mathcal{C}$ attains maximum sum-distance iff it satisfies parity balance across the $16$ HCWs (Lemma \ref{lem:maxsum}). Let $n_{\p}$ denote the count of HCWs assigned parity $\p \in \{00, 01, 10, 11\}$. Parity balance at each bit position requires:
\begin{equation}
\label{eq:rig_bal}
    n_{00} + n_{01} = n_{10} + n_{11} = 8,\ n_{00} + n_{10} = n_{01} + n_{11} = 8
\end{equation}
Theorem~\ref{thm:validity} mandates that for any HCW pair $(\mathbf{c}_i, \mathbf{c}_j)$ such that $\dH(\mathbf{c}_i, \mathbf{c}_j) = 3$, $\p(\mathbf{c}_i) \neq \overline{\p(\mathbf{c}_j)}$. Let $G_3$ be the $7$-regular bipartite graph with partite sets $\mathcal{C}_e$ and $\mathcal{C}_o$ (Proposition \ref{prop:bipartite}).

Let $\p_{\mathcal{C}}(\mathcal{C}_e)$ and $\p_{\mathcal{C}}(\mathcal{C}_o)$ denote the set of parity vectors assigned by $\mathcal{C}$ to the HCWs in $\mathcal{C}_e$ and $\mathcal{C}_o$, respectively. Suppose $\p_{\mathcal{C}}(\mathcal{C}_e)$ contains a non-complementary pair of parities, say $\p$ and $\mathbf{q}$, $\mathbf{q} \neq \overline{\p}$. W.l.o.g, let $\p=00$ and $\mathbf{q}=01$. Each HCW in $\mathcal{C}_e$ with the parity assignment $00$ forbids its $7$ neighbors in $\mathcal{C}_o$ from being assigned $11$. Similarly, each HCW in $\mathcal{C}_e$ assigned $01$ forbids its $7$ neighbors in $\mathcal{C}_o$ from being assigned $10$. Given $|\mathcal{C}_o|=8$, $7+7-8=6$ nodes in $\mathcal{C}_o$ are simultaneously forbidden from being assigned the parities $\{10, 11\}$. These $6$ nodes must therefore be assigned parities from $\{00, 01\}$.


This creates a surplus of zeros in the first parity bit position. Specifically, if $6$ nodes in $\mathcal{C}_o$ have a leading '$0$', then to satisfy \eqref{eq:rig_bal}, $\mathcal{C}_e$ must have a corresponding surplus of ones. However, populating $\mathcal{C}_e$ with parities starting with '$1$' ($\{10, 11\}$) triggers the same banning mechanism in reverse, forbidding $\{00, 01\}$ from $\mathcal{C}_o$. This circular constraint prevents the system from reaching the $8$-$8$ distribution required by \eqref{eq:rig_bal}. 
Consequently, parity balance and validity (Theorem \ref{thm:validity}) can only coexist if the HCWs in each partite set are assigned parities exclusively from a single complementary pair. Let the HCWs in $\mathcal{C}_e$ and $\mathcal{C}_o$ be assigned parities from $\{\p, \overline{\p}\}$ and $\{\mathbf{q}, \overline{\mathbf{q}}\}$, respectively. Letting $n_{\p}=x, n_{\overline{\p}}=8-x$ in $\mathcal{C}_e$ and $n_{\mathbf{q}}=y, n_{\overline{\mathbf{q}}}=8-y$ in $\mathcal{C}_o$, the balance equations \eqref{eq:rig_bal} reduce to $x+y=8$ and $x+(8-y)=8$. This yields the unique solution $x=y=4$, defining the $4/4$ bipartite split.
Since the construction described in section \ref{subsubsec:Construction} is the only way to achieve parity balance needed for maximum sum-distance, all optimal SEFCCs that attain maximum sum-distance follow this bipartition-based parity-assignment.  
\end{IEEEproof}

\begin{theorem}[Max-sum Max-$d_{\min}$ FCCs]
\label{thm:maxsum_dmin}
    An optimal SEFCC $\mathcal{C}: \F_2^7 \rightarrow \F_2^9$ for the HCMF $f_H$ satisfies all of the following distance properties:
    \begin{enumerate}
        \item[(a)] it achieves the maximum theoretical sum-distance
    $\Sigma_d(\mathcal{C}) = 73728$ characterized in Lemma~\ref{lem:maxsum};
        \item[(b)] it achieves the maximum possible minimum distance, namely
    $d_{\min}(\mathcal{C}) = 2$; and 
        \item[(c)] it minimizes the number of codeword pairs in $\mathcal{C}$ that are at Hamming distance $2$;
    \end{enumerate}
    if and only if $\mathcal{C}$ is obtained via the bipartition-based code construction described in section \ref{subsubsec:Construction}. 
\end{theorem}
 \begin{IEEEproof}
Conditions (a) and (b), namely maximum sum-distance $\Sigma_d(\mathcal{C})$ and $d_{\min}(\mathcal{C}) = 2$, follow directly from Theorems~\ref{thm:maxsumFCC} and \ref{thm:dmin2}, respectively. Since the bipartite construction restricts the partite sets $\mathcal{C}_o$ and $\mathcal{C}_e$ to disjoint parity sets $P_1 = \{00, 11\}$ and $P_2 = \{01, 10\}$, no pair $(\mathbf{c}_i, \mathbf{c}_k)$ with $\dH(\mathbf{c}_i, \mathbf{c}_k) = 3$ can be assigned the same parity. By Theorem~\ref{thm:dmin2}, this ensures $d_{\min}(\mathcal{C}) = 2$.

To prove (c), we minimize the number of pairs $(\Enc(\mathbf{u}_i), \Enc(\mathbf{u}_j))$ at Hamming distance $2$. A distance-2 pair in the FCC arises when the sum of message distance and parity distance equals 2:
    \(\dH(\mathbf{u}_i, \mathbf{u}_j) + \dH(\p_\mathcal{C}(\mathbf{u}_i), \p_\mathcal{C}(\mathbf{u}_j)) = 2.\)
These pairs originate from three sources:
\begin{enumerate}
    \item[(i)] \textbf{Intra-sphere pairs:} Two NHCWs $\mathbf{u}_{i,a}, \mathbf{u}_{i,b}$ within the same sphere $S_i$. Here $\dH(\mathbf{u}_{i,a}, \mathbf{u}_{i,b})=2$ and the parity distance is $0$ as both share $\overline{\p_\mathcal{C}(\mathbf{c}_i)}$. There are $16 \times \binom{7}{2}$ such pairs, invariant under any valid assignment.
    \item[(ii)] \textbf{Inter-sphere boundary pairs:} NHCWs $\mathbf{u}_{i,a} \in S_i$ and $\mathbf{u}_{k,b} \in S_k$ where $\dH(\mathbf{c}_i, \mathbf{c}_k)=3$. For these, $\dH(\mathbf{u}_{i,a}, \mathbf{u}_{k,b})=1$ and the parity distance $\dH(\overline{\p_\mathcal{C}(\mathbf{c}_i)}, \overline{\p_\mathcal{C}(\mathbf{c}_k)})=1$ (since the parities are from disjoint sets $P_1, P_2$). This count is a constant for any bipartite construction.
    \item[(iii)] \textbf{Identical-parity pairs:} NHCWs $\mathbf{u}_{i,a} \in S_i$ and $\mathbf{u}_{k,b} \in S_k$ where $\dH(\mathbf{c}_i, \mathbf{c}_k)=4$ and $\p_\mathcal{C}(\mathbf{c}_i) = \p_\mathcal{C}(\mathbf{c}_k)$. Here $\dH(\mathbf{u}_{i,a}, \mathbf{u}_{k,b})=2$ and the parity distance is $0$.
\end{enumerate}
To minimize the count in category (iii), we must minimize the number of HCW pairs at distance 4 that share a parity. In a partite set of 8 HCWs, every node is at distance 4 from other nodes in that partite set. If $x$ nodes in $\mathcal{C}_e$ are assigned $\p \in P_2$ and $8-x$ are assigned $\overline{\p}$, the number of same-parity pairs is $\binom{x}{2} + \binom{8-x}{2}$. This quadratic is minimized if and only if $x = 8-x = 4$. The same logic applies to $\mathcal{C}_o$ and $P_1$. Any deviation from this $4/4$ split increases the number of distance-2 pairs, thereby worsening the distance spectrum.
\end{IEEEproof}

It is well known that the probability of error of a code is governed by its distance spectrum: codes with larger minimum distance exhibit superior error performance, and among codes with identical minimum distance, those with fewer codeword pairs at the minimum distance perform better, with higher-distance terms contributing subsequently. 

For SEFCCs designed for the HCMF $f_H$, Theorem~\ref{thm:maxsum_dmin} establishes that maximizing the sum-distance is equivalent to simultaneously maximizing the minimum distance $d_{\min}$ and minimizing the number of codeword pairs at this $d_{\min}$. Consequently, SEFCCs that achieve the maximum sum-distance possess the most favorable distance spectrum permitted under the FCC constraints. This implies that, in the present setting, sum-distance maximization is sufficient to guarantee optimal data-bit error performance in the distance-spectrum sense, and is therefore not merely a heuristic—as is often the case in general code design—but an exact optimality criterion.

\section{Simulation Results}
\label{sec:Sim}
We evaluate the FER as well as BER of two distinct SEFCCs for the HCMF over an AWGN channel. The code $\mathcal{C}_1$ is obtained via the bipartite construction in Section \ref{subsec:CodeConst} and the code $\mathcal{C}_2$ is the one that results in optimal FER as given in Theorem \ref{thm:optimFER}. We assume that the coded bits are transmitted using BPSK modulation and that decoder performs soft-decision decoding. All simulations utilize the specific $[7,4,3]$-Hamming code defined in Example \ref{ex:HC_code}.


Matching theoretical guarantees, the simulated FER performance of $\mathcal{C}_2$ is superior to that of $\mathcal{C}_1$, whereas the simulated BER performance of $\mathcal{C}_1$ is superior to that of $\mathcal{C}_2$, since $\mathcal{C}_2$ attains optimal function protection despite having $d_{\min}=1$, while $\mathcal{C}_1$ achieves $d_{\min}=2$ and an optimized distance spectrum that provides stronger data-bit protection.


\begin{figure}[h]
    \centering
    \includegraphics[width=0.98\columnwidth]{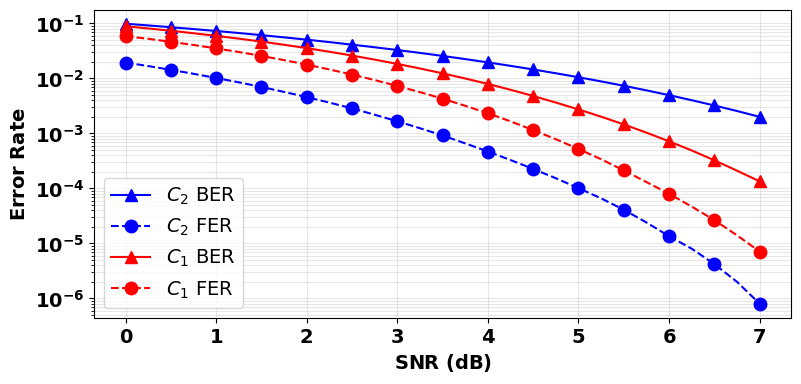}
    \caption{BER and FER comparison of $\mathcal{C}_1$ and $\mathcal{C}_2$.\label{fig:berfer}}
\end{figure}




\section{Conclusion}
\label{sec:Conc}
We studied SEFCC constructions for the $[7,4,3]$ HCMF and provided a complete characterization of valid parity assignments via a distance-$3$ constraint. Leveraging the induced bipartite geometry of the Hamming code, we identified a distance-maximizing parity-balance condition and derived a full class of constructions that simultaneously maximize sum-distance, achieve the largest possible minimum distance, and minimize the number of minimum-distance codeword pairs. These results demonstrate that, for the HCMF SEFCC problem, sum-distance maximization is not merely heuristic but exactly enforces the distance-spectrum properties that govern data-bit error performance. While the analysis is presented for the $[7,4,3]$ case, the underlying geometric and distance-spectrum principles suggest natural extensions to SEFCC design for general Hamming codes.


\end{document}